\documentclass[11pt,letterpaper]{article}

\usepackage{amsmath,amssymb,amsthm}
\usepackage{booktabs} 
\usepackage{xspace}
\usepackage[pagebackref=true,hidelinks]{hyperref}
\usepackage{accents} 
\usepackage{pgf} 
\usepackage{xparse} 
\usepackage{bbm}
\usepackage{enumitem}
\usepackage[numbers]{natbib}
\usepackage{verbatim}
\usepackage[margin=1in]{geometry}

\usepackage[ruled,linesnumbered]{algorithm2e} 

\SetAlFnt{\small}
\SetAlCapFnt{\small}
\SetAlCapNameFnt{\small}
\SetAlCapHSkip{0pt}
\IncMargin{-\parindent}

\begin{document}

\newcommand{\alginit}{}
\newcommand{\alg}{\textsc{ALG}}

\newcommand{\nsi}[1]{{\color{red}{[NSI: #1]}}}
\newcommand{\blcomment}[1]{{\color{blue}{[BL: #1]}}}
\newcommand{\bledit}[1]{{\color{blue}{#1}}}

 \newtheorem{theorem}{Theorem}[section]
 \newtheorem{corollary}[theorem]{Corollary}
 \newtheorem{lemma}[theorem]{Lemma}
 \newtheorem{claim}[theorem]{Claim}
\newtheorem{fact}[theorem]{Fact}
 \newtheorem{proposition}[theorem]{Proposition}
 \newtheorem{conjecture}{Conjecture}
 \newtheorem{property}{Property}
\newtheorem{observation}[theorem]{Observation}
 \newtheorem{remark}{Remark}
 \newtheorem{definition}{Definition}
 \newtheorem{example}{Example}

\title{Dynamic Weighted Matching with Heterogeneous Arrival and Departure Rates}


\author{Natalie Collina\thanks{Harvard University, \url{nataliecollina@gmail.com}}%
\and Nicole Immorlica\thanks{Microsoft Research, \url{nicimm@microsoft.com}} \and Kevin Leyton-Brown\thanks{University of British Columbia, \url{kevinlb@cs.ubc.ca}} \and Brendan Lucier\thanks{Microsoft Research, \url{brlucier@microsoft.com}} \and Neil Newman\thanks{University of British Columbia, \url{newmanne@cs.ubc.ca}}}

\date{}

\maketitle

\begin{abstract}
We study a dynamic non-bipartite matching problem.  There is a fixed set of agent types, and agents of a given type arrive and depart according to type-specific Poisson processes.  Agent departures are not announced in advance.  The value of a match is determined by the types of the matched agents.  We present an online algorithm that is (1/8)-competitive with respect to the value of the optimal-in-hindsight policy, for arbitrary weighted graphs.  Our algorithm treats agents heterogeneously, interpolating between immediate and delayed matching in order to thicken the market while still matching valuable agents opportunistically.
\end{abstract}

\section{Introduction}

Matching markets are ubiquitous in online platforms.  Sponsored search auctions like Google Adwords match ads and users, ridesharing systems like Uber and Lyft match drivers and riders, online markets like Amazon and eBay match sellers and buyers.  In each case, the value of a match is a function of the types of participating agents.  In sponsored search auctions, a restaurant ad is more valuable when matched to a geographically co-located user.  In ridesharing systems, a driver and rider have higher utility for being matched to each other if they are nearby. In an online market, buyers might have heterogeneous preferences over service/product quality and price trade-offs which impact match quality.  

The role of the platform is to find high-value matches.  However, this task is significantly complicated by the fact that agents arrive and depart dynamically over time, and may fail to inform the platform of their departure. In this paper, we mitigate this complication by assuming that agents have known Poisson arrival and departure rates that are a function only of their type.  This allows us to characterize the optimal expected value from matches using a linear program. This program bounds the rate at which each pair of types match to one-another in the optimal solution.  Our algorithm uses these  LP-based estimates of the optimal rates as guidelines.
When an agent arrives to the market, we use these guidelines to choose with which other types of agents (if any) a match should be attempted.  If the newly-arriving agent is not successfully matched to any other, then the agent is added to a waiting pool of future match opportunities.
We prove the resulting algorithm is a constant approximation to the optimal-in-hindsight policy, with competitive ratio at most $8$.  While we motivate our problem in the context of bipartite matchings, we note our solution holds for general non-bipartite graphs.

There is a significant body of prior literature on dynamic stochastic matching in settings where agent departures are immediate or deterministic (and hence predictable)~\cite{FeldmanMMM09,BernhardMZ11,MahdianY2011,huang2018match,huang2019tight}, or where the platform is informed immediately before an agent departs~\cite{Akbarpour2017,AshlagiBDJSS19,dickerson2018assigning,truong2019prophet}.  In such settings, it is natural for the platform to delay matches until an agent is about to depart, in order to maximize the set of available options.  In contrast, when the platform cannot predict departures, there is a tension between taking a guaranteed (but potentially suboptimal) match now, or pushing one's luck to see if a better match arrives later.  The main technical challenge in developing an online policy is navigating this tradeoff for agents of different types.

Our LP-based approach is certainly not new in the context of stochastic matching, but we find that our result has several interesting qualitative insights, especially for settings where agent departures are random, heterogeneous, and unannounced.
First, our algorithm treats matches heterogeneously.  For some matches, the linear program suggests forming them at a high rate.  Our algorithm treats these matches as a greedy algorithm would, matching them (almost) immediately upon arrival.  For other matches, the linear program suggests forming them at a low rate.  Our algorithm treats these matches more like a periodic clearing algorithm would, allowing the market to thicken before attempting the matches.


This heterogeneous treatment is important for good approximations in our setting.  Consider, for example, an environment with two types of buyers, low and high, and one type of seller.  The low buyers arrive frequently to the market and depart at a constant rate, whereas the sellers arrive less often.  The high buyers arrive much less frequently than the sellers, and depart immediately after they arrive, but matches involving these high buyers account for almost all the value of the optimal policy.  In this case, it is important to greedily match the high buyers and delay matches with the low buyers to thicken the market.  A uniformly greedy policy, that immediately matches all agents, will likely have no sellers in the market when high buyers arrive, as there are always low buyers available to match with them. A periodic clearing algorithm that attempts to thicken the market by delaying all matches for a fixed period of time will likely have no access to high buyers at match time, since high buyers depart immediately after they arrive.  

Another qualitative insight of our result is the importance of being conservative in matching attempts. Our algorithm scales back the match-rate estimates of the linear program by $50\%$.
This might seem incredibly wasteful at first blush.  However, this scaling is provably necessary: 
if one were to remove this scaling from our algorithm
then it would not achieve any bounded approximation to the optimal matching.  Intuitively, the issue is that the matching policy must leave some slack in the system --- by leaving a certain fraction of agents unmatched --- in order to take advantage of unexpected fortuitous events where a very valuable match becomes possible.  Since an optimal LP solution typically would leave no such slack, one can instead guarantee it by being conservative when matching.

As is common in the dynamic stochastic matching literature, our approach is to solve an LP relaxation of the offline optimal matching problem, then use this solution as guidance for our online matching policy.  We prove that the resulting policy obtains a constant approximation to the LP benchmark, which is only stronger than the offline optimal match value (and hence the optimal online policy).  
The main technical hurdle is that the outcome of these matching attempts is determined by the state of which types of agents are present in the market, and this introduces correlations across time.  For instance, whether a certain type of agent is present in the market is (negatively) correlated with the presence of other agents that generate high value from matches with it.  In principle, such correlations could result in scenarios where a certain type is either not present at all or is overabundant, impeding our ability to track the LP relaxation which is smooth across time.  In our analysis we show that the impact of such correlations is bounded, by coupling the availability of agents in the system with independent Poisson processes that dominate (or are dominated by) them.

\subsection{Related Literature}
\label{sec:related}

There is a vast recent literature on algorithms for online matching (sometimes called online task arrival).  In a seminal paper, Karp et al.~\cite{KarpVV90} 
consider an (unweighted) online bipartite matching problem where one side of the graph is static and the vertices of the other side arrive online.  They show that a randomized greedy matching method obtains a $(1 - 1/e)$ approximation and that this is tight.  This was later extended by Mehta et al.~\cite{MehtaSVV05} to a generalized weighted matching environment motivated by ad auctions, with budget constraints on the static side of the market.  Both of these results assume adversarial types.  

Stochastic variants of the online bipartite matching problem have been studied as well. Feldman et al.~\cite{FeldmanMMM09} consider a stochastic variant in which vertex types on the online side of the market are drawn i.i.d.\ from a fixed distribution.  They showed how to beat the adversarial bound of $(1 - 1/e)$ in this stochastic setting, using an LP-based approach that solves for a fractional (expected) matching, then rounds online using a flow decomposition.  This led to 
a sequence of papers that improved the approximation factors for both the weighted and unweighted versions of the stochastic problem
~\cite{BernhardMZ11,MahdianY2011}, including variants with stochastic rewards~\cite{MehtaP12,MehtaWZ2015} and with capacities on the fixed side~\cite{AlaeiMV2012}.  Gravin and Wang~\cite{Gravin2019prophet} obtain a constant approximation for a related variant inspired by prophet inequalities, where edges (rather than nodes) arrive online and must be matched immediately or lost.

Our model is closer in spirit to the literature on dynamic matching, where agents on both sides of the market arrive and depart over time. An algorithm proposes matches online between agents that are simultaneously present. 
Huang et. al~\cite{huang2018match} study an unweighted model in which node arrivals and departures are adversarial, but nodes announce when they are about to depart. They derive constant competitive online algorithms; in a later paper, Huang et al.~\cite{huang2019tight} find tight competitive ratios. 
Akbarpour et al.~\cite{Akbarpour2017} similarly consider an unweighted version in which agents depart at arbitrary times and inform the market when they are about to depart, but arrivals are stochastic.  In this case, it is approximately optimal to match agents as they go critical.  On the other hand, they show that without departure warnings, greedily matching agents as they arrive is nearly optimal.  
As the graph is unweighted in their model and agents are homogenous, analysis can proceed by studying the limiting distribution of the number of agents in the market.  

The case of weighted matching with departure warnings was studied by Ashlagi et al.~\cite{AshlagiBDJSS19}, and they obtain a constant approximation to the optimal weighted matching.  
When agents on both sides arrive according to a known IID random process, Dickerson et al.~\cite{dickerson2018assigning} provide constant competitive algorithms under the assumption that one side (say workers) never depart until they are assigned, and the other side (say tasks) depart immediately after arrival if unassigned. 
Truong and Wang~\cite{truong2019prophet} consider a related weighted bipartite matching model where agents arrive according to a general stochastic process, agents on one side depart after a fixed deterministic amount of time, agents on the other side depart immediately after arrival if unassigned, and they likewise obtain constant competitive algorithms.  Importantly, in all of these works it is assumed that the platform knows when an agent is about to leave the system, either because this can be perfectly predicted or because the platform is explicitly notified, and the platform can therefore wait until an agent ``goes critical'' before attempting a match.  In contrast to these works, we assume the platform is not notified of (and cannot predict) impending departures.

Independently and concurrently with our work, Aouad and Saritac~\cite{Aouad20} studied a similar model of dynamic matching with unannounced departures.  They likewise find that there is a tension between greedy matching and batching.  They develop an online algorithm guided by a quadratic program, and show that it is ($4e/(e-1)$)-competitive for arbitrary compatibility graphs.  In contrast, our method is based on linear programming (rather than quadratic programming), and our competitive ratio bound is weaker ($8$ versus $4e/(e-1)$).  They also study a cost-minimization version of the problem, for which they develop an online algorithm that they analyze theoretically and evaluate on empirical data.  We leave open the question of whether a combination of the ideas in these works could be used to develop algorithms with improved competitive ratio.

Other papers consider the related problem of minimizing average waiting time.  Anderson et al.~\cite{Anderson2017} find that matching agents as they arrive is nearly optimal even with departure warnings.  Ashlagi et al.~\cite{ashlagi2019matching} consider a model with two agent types -- hard-to-match and easy-to-match -- and derive structural insights about policies that miminize average waiting time.  Baccara et al.~\cite{baccara2018optimal} consider a hybrid model with two agent types in which agents have varying match values and also incur waiting costs (but never leave the system).





\section{Preliminaries}
\label{sec:model}

We consider a model with agents that arrive and depart over time.  The type space of agents is $X$. Agents of type $x \in X$ arrive according to a homogeneous Poisson point process of rate $\lambda_{x} > 0$.\footnote{We discuss Poisson processes more formally in Section~\ref{sec:model.poisson}.}  Each agent of type $x$ that arrives then departs at Poisson rate $\mu_{x} > 0$.  We will allow $\mu_x$ to be $\infty$, which indicates that an agent of type $x$ always departs immediately after arriving.  For an agent $i$ of type $x$, we will write $a_i$ and $d_i$ for its realized arrival and departure times, respectively.  Throughout, we refer to types of agents with letters $x$ and $y$, and to specific agents with letters $i$ and $j$.

A {\em matching} is a set $\tau$ of times and a pair of matched agents for each time $t\in\tau$.  A matching is {\em feasible} if, for all matching times $t\in\tau$, the agents matched at $t$ a) have already arrived and not yet departed, and b) have not been matched to anyone else at or before time $t$.  The value of matching an agent of type $x\in X$ to an agent of type $y\in X$ is $v_{xy}$.  For convenience, we sometimes denote the total value of all matches made at time $t$ by $v_t$.

A {\em matching policy} chooses, at each time $t$, based only on the history up until time $t$, whether to match a pair of agents or to make no match.  A {\em policy with hindsight} can revise past decisions, whereas for an {\em online policy}, all decisions are irrevocable.  For any policy and time $T$, let $\tau(T)$ be all times $t \leq T$ at which it made a match,\footnote{Note for an online policy, $\tau(T)\subseteq\tau(T')$ whenever $T\leq T'$; however this need not hold for a policy with hindsight.} and $v_t$ be the value of the matches made at time $t$, if any.  Then the value of the policy is:
$$\liminf_{T\rightarrow\infty}\frac{1}{T}\cdot E\left[\sum_{t\in\tau(T)}v_t\right]$$
where the expectation is over the randomness in the arrival/departure process as well as any randomness in the policy.  That is, the policy's value is the long-run average value of matches made per unit of time.

\subsection{Poisson Processes}
\label{sec:model.poisson}


We now describe Poisson processes more formally.  A point process is a random countable set of points $Z = \{z_1, z_2, \dotsc\}$.  We restrict attention to the case $Z \subset R_{\geq 0}$, where we can interpret $Z$ as a collection of event times.  We refer to a point process by its set of points $Z$, which we think of as a random variable.  Given $t \geq 0$, we write $Z(t)$ for the event that $t \in Z$.

For any $T \geq 0$, we'll write $n_Z(T)$ for the number of points in $Z \cap [0,T]$; we think of this as the (random) number of events that occur before time $T$.  Given two point processes $Z$ and $Y$, we'll say that $Z$ stochastically dominates $Y$ if there is a coupling between $Z$ and $Y$ such that, for each $T > 0$, $\Pr[ Y \subseteq Z ] = 1$.

A Poisson point process with intensity function $\lambda(t)$ is a point process such that
\begin{enumerate}
    \item the set of points in any two disjoint intervals are independent, and 
    \item the number of points in any given interval $(a,b]$, with $a \leq b$, follows a Poisson random variable with parameter (mean) $\int_a^b \lambda(t) dt$.
\end{enumerate}
When $\lambda(t)$ is a constant function, say $\lambda(t) = \lambda$, then we say the Poisson point process is homogeneous with rate $\lambda$.
%
The following standard facts about Poisson processes will be helpful in our analysis.   

\begin{fact} 
\label{fact:expectations}
Given a homogeneous Poisson process $Z$ of rate $\lambda$, write $n_Z(T)$ for the number of events that occur before time $T$. Then $E[n_Z(T)] = T \lambda$.  Moreover, $\lim_{T \to \infty} n_Z(T)/T$ exists and equals $\lambda$ with probability $1$.
\end{fact}

\begin{fact} 
\label{fact:recent}
Suppose we have independent homogeneous Poisson processes $Z_1, \dotsc, Z_n$ of rates $\lambda_1, \dotsc, \lambda_n$ respectively. Then the probability that the earliest event (i.e., minimum point) in $\cup Z_i$ lies in $Z_i$ is $\lambda_i / (\sum_k \lambda_k)$.
\end{fact}

\begin{fact} 
\label{fact:combine}
Suppose $Z$ is a Poisson process with intensity function $\lambda(t)$, and $Z'$ is a random set generated by adding each $t \in Z$ to $Z'$ independently with probability $p(t)$.  Then $Z'$ is a Poisson point process, with intensity function given by $\lambda'(t) = \lambda(t)p(t)$.
\end{fact}

A corollary of Fact~\ref{fact:combine} is that if $Z$ is a homogeneous Poisson process of rate $\lambda$ and $Z'$ is a homogeneous Poisson process of rate $\lambda' < \lambda$, then $Z$ stochastically dominates $Z'$.  This is because we can couple $Z$ and $Z'$ by first realizing $Z$, then adding each element of $Z$ to $Z'$ independently with probability $\lambda' / \lambda$.

\section{An Upper Bound}


We construct an online policy whose value is a constant fraction of the optimal-in-hindsight policy.   To do so, we first develop a linear-programming (LP) upper bound on the value of the optimal-in-hindsight policy for large time horizons.\footnote{Taking the limit as the time horizon grows allows us to ignore lower-order terms.}  The value of the optimal solution is the expectation over the randomness in arrivals and departures of instance-optimal solutions, and so can be written as the expectation of the sum of match values.  We then transform this LP into one with strictly more constraints but the same optimal value.  We will use this second LP to develop and analyze our online policy.

In the following LP, the variable $\alpha_{xy}$ is the fraction of nodes of type $y$ which match to preexisting nodes of type $x$, when considered over all arrivals of agents of type $y$.

\begin{alignat}{5}
& \mbox{\textbf{LP-UB:}} &\quad & \mbox{maximize} &\quad & \displaystyle\sum\limits_{x,y \in X} v_{xy}\alpha_{xy}\lambda_{y} \nonumber\\
& & & \mbox{subject to} & & \displaystyle \alpha_{xy} \leq  \frac{\lambda_{x}}{\mu_{x}}  &\quad &\forall x,y \in X \label{eq.LP.cond1}\\
& & & & & \displaystyle\sum\limits_{y \in X}\alpha_{xy}\lambda_{y} + \sum\limits_{y \in X}\alpha_{yx}\lambda_{x} \leq  \lambda_{x}  &\quad &\forall x \in X \label{eq.LP.cond2}\\
& & & & & \alpha_{xy} \in [0,1], &\quad &\forall x,y \in X \label{eq.LP.range}
\end{alignat}


Constraint~\eqref{eq.LP.cond1} bounds the fraction of the time that some node of type $y$ matches to some previously arrived node of type $x$ by the probability that a node of type $x$ is present in the system at any given time. Constraint~\eqref{eq.LP.cond2} bounds the total rate at which a type can match by the total rate at which the type arrives.  On the left-hand-side, the first sum  captures the rate at which a type matches to those arriving after it; the second sum captures the rate at which a type matches to those who arrived before it. Note that constraints~\eqref{eq.LP.cond2} and~\eqref{eq.LP.range} together imply that $\sum_{x \in X}\alpha_{xy} \leq 1$ for all $y$. This makes intuitive sense: the total fraction of the time that a node matches to any preexisting type cannot be greater than $1$. 

We will first demonstrate that the value of LP-UB represents an upper bound on the expected value of the max-weight offline matching. 

\begin{lemma}
\label{lem:ub}
Let $v^*$ be the optimal value of LP-UB.  Then the value of any matching policy, including policies with hindsight, is at most $v^*$.
\end{lemma}

To prove Lemma~\ref{lem:ub}, we
consider the set of agents who arrive up to some time $T$, and interpret the constraints of LP-UB as conditions on matchings in the induced graph of potential matches.  These finite conditions include lower order terms, but these disappear when taking the limit as $T$ grows large.

\begin{proof}
For each $T > 0$, let $G_T$ be a random weighted graph whose nodes are the agents that arrive before time $T$.  Write $N_x(T)$ for the set of nodes of type $x$.  Recall that for a given node $i$, we write $a_i$ for the time that $i$ arrives and $d_i$ for the time that it departs.  

Two agents $i$ and $j$ (of types $x$ and $y$ respectively) can be matched if $[a_i, d_i] \cap [a_j, d_j] \neq \emptyset$, in which case the weight of the match is $v_{xy}$.  The value of the maximum-weight matching in $G_T$ is an upper bound on the value obtained by any matching policy over time interval $[0,T]$.

Write $M(T)$ for the maximum-weight matching in $G_T$, and write $V(T)$ for its weight. Then the value of any dynamic matching policy is at most
\[ \liminf_{T \to \infty} \frac{1}{T} \cdot E[V(T)], \]
where the expectation is over the realization of the graphs $G_T$.  So it suffices to show that this quantity is at most $v^*$.


Fix some choice of $T$.  
For $x \in X$, write $n_x(T) = |N_x(T)|$ for the random variable denoting the number of nodes of type $x$ that arrive in $[0,T]$.  
For $x,y \in X$, write $\Psi_{xy}(T)$ for the number of nodes of type $y$ that (a) are matched (in $M(T)$) to a node of type $x$, and (b) arrived later than the node to which they matched.  Define $\alpha_{xy} = \liminf_{T \to \infty} \frac{1}{T}E[\Psi_{xy}(T)]/\lambda_y$, where the expectation is over the realization of the graphs $G_T$.  

Note that 
$V(T) = \sum_{x,y} \Psi_{xy}(T) v_{xy}$,
the total value of all matches made in $M(T)$.  Taking expectations, dividing by $T$ and taking limit inferiors gives
\[ \liminf_{T \to \infty} \frac{1}{T} E[V(T)] = \sum_{x,y} v_{xy} \liminf_{T \to \infty} \frac{1}{T} \cdot E[\Psi_{xy}(T)]=\sum_{x,y}v_{xy}\lambda_y\alpha_{xy}.  \]

Since each node of type $x$ can match at most once, and matches either to a node that arrived before it or after it, we have
\[ n_x(T) \geq \sum_{y \in X}  \Psi_{yx}(T) + \sum_{y \in X}  \Psi_{xy}(T). \]
Dividing by $T$, taking expectations, and taking a limit inferior on both sides yields
\[ \liminf_{T \to \infty}\frac{1}{T} \cdot E[n_x(T)] \geq \liminf_{T \to \infty}\frac{1}{T}\cdot \left( \sum_{y \in Y} E[ \Psi_{yx}(T)] + \sum_{y \in X} E[ \Psi_{xy}(T)] \right). \]
Since $\lim_{T \to \infty} n_x(T)/T = \lambda_x$ with probability $1$ (by Fact~\ref{fact:expectations}), and by definition of $\alpha_{xy}$ and $\alpha_{yx}$, we have
\[ \lambda_x \geq \sum_{y \in Y} \lambda_x \alpha_{yx} + \sum_{y \in Y} \lambda_y \alpha_{xy}. \]
That is, the values $\alpha_{xy}$ satisfy condition~\eqref{eq.LP.cond2} of LP-UB.

Next, define $m_{yx}(T)$ to be the number of nodes of type $x \in X$ such that, at the time that the node arrives, there is at least one node of type $y$ present.  Since $\Psi_{yx}(T)$ counts each node of type $x$ that matches to a node of type $y$ that is present when the former arrives,
\[ \Psi_{yx}(T) \leq m_{yx}(T). \]
Consider now the expected value of $m_{yx}(T)$ over the realization of graph $G_T$.
As the arrival process is independent across types and uniform across time, $m_{yx}(T)$ is precisely $n_x(T)$ times the probability that there is at least one node of type $y$ present at a uniformly random time in $[0,T]$.  This latter probability is at most $\frac{1}{T} \cdot E[\sum_{i \in N_y(T)}(d_i - a_i)]$, the total sum of type-$y$ node lifetimes as a fraction of $T$.  Thus  
\[ E[\Psi_{yx}(T)] \leq E[m_{yx}(T)] \leq E\left[n_x(T) \cdot \frac{\sum_{i \in N_y}(d_i - a_i)}{T}\right]. \]
Since the arrivals and departures of different nodes in $N_y$ are independently and identically distributed, we have
\[ E\left[n_x(T) \cdot \frac{\sum_{i \in N_y}(d_i - a_i)}{T}\right] = E[n_x(T)] \cdot \frac{E[n_y(T)]}{T \mu_y}  = \frac{T \lambda_x \lambda_y}{\mu_y}, \]
where the first equality follows since the expected amount of time a node of type $y$ spends in the system is $1/\mu_y$ by the Poisson process, 
and the second equality follows by Fact~\ref{fact:expectations}. We conclude that
\[ E[\Psi_{yx}(T)] \leq T \lambda_x \lambda_y / \mu_y \]
and hence
\[ \liminf_{T \to \infty} \frac{1}{T} \cdot  E[\Psi_{yx}(T)] \leq \lambda_x \lambda_y / \mu_y. \]
By the definition of $\alpha_{yx}$, this implies
\[ \alpha_{yx} \leq \lambda_y / \mu_y. \]
That is, the values $\alpha_{xy}$ satisfy condition~\eqref{eq.LP.cond1} of LP-UB.

The limit of the expected match value is therefore the value of a feasible solution to LP-UB, and in particular is at most the optimal value of LP-UB.
\end{proof}

\section{Online Matching Policy}
\label{sec:policy}

We now present our online matching policy, \textsc{OnlineMatch}.
Our policy first solves LP-UB in advance of any arrivals, and then uses the solution to guide its matching decisions.  As demonstrated in the previous section, the solution to the LP-UB should be thought of as describing the optimal matching rates between types, subject to constraints that hold as time approaches infinity. Our goal is to create a policy that approximately matches the value of this LP, which we will achieve by obtaining a constant approximation to these matching rates. 

\begin{algorithm}[t]
\alginit
\caption{Algorithm \textsc{OnlineMatch}
\label{alg.avail}}
\setcounter{AlgoLine}{0}
	\SetKwInOut{Require}{require}
	\SetKwInOut{Input}{input}
	\Require{Scaling parameter $\gamma \in (0,1]$}
	\Input{Online arrivals of agents}
	
	\BlankLine
	\label{alg.line1}$(\alpha_{xy}) := $ Solution to LP-UB\;
	\For{\label{alg.line2}each agent $i$ arriving at time $t$, say of type $y \in X$}{
    	\For{\label{alg.line3}each type $x \in X$ in a uniformly random order}{
            \uIf{\label{alg.line4}there is at least one unmatched agent $j$ of type $x$ in the market}{match $i$ and $j$ with probability $\gamma \cdot \alpha_{xy} \cdot \max\left(1,\frac{\mu_x}{\lambda_x}\right)$}\label{alg.line5}
        }
	}
\end{algorithm}

Suppose that an agent, say agent $i$ of type $y$, arrives at time $t$.  The algorithm will then iterate through all types (including $y$) in a fixed order (line~\ref{alg.line3}).  For each considered type $x$, if there are any agents of type $x$ present and unmatched in the market, the algorithm will select one of them arbitrarily and attempt to match it with agent $i$.  With probability 
$\gamma \cdot \alpha_{xy} \cdot \max\left(1,\frac{\mu_x}{\lambda_x}\right)$ 
the match occurs, in which case the algorithm completes and awaits the next agent arrival.  Otherwise, the algorithm moves on to the next type in $X$.  If agent $i$ is not matched after every $x \in X$ has been considered, then we leave agent $i$ unmatched and await the next arrival.

The match probability on line~\ref{alg.line5} deserves some discussion.  This probability depends on the solution to LP-UB, and is the mechanism by which the algorithm attempts to follow the matching rates proposed by the LP.  One might be tempted to simply use $\alpha_{xy}$ as the match probability.
However, when constructing an online policy we must consider the difference between unconditional match rates and matching rates conditional on agent types being present in the market.  It may be that a particular type is extremely unlikely to be present to match during a given attempt. 
Consider a problem instance that includes a type $x$ with arrival rate 1 and departure rate $1/\epsilon$, and a corresponding LP solution where $\alpha_{xy} = \epsilon$ for some $y$ (note that this does not immediately violate any constraints, as the upper bound on $\alpha_{xy}$ could be as high as $\epsilon$). The probability that any agent of type $x$ will be present when an agent of type $y$ arrives is at most $\epsilon$ (see Lemma~\ref{lem:present}). Thus an online policy that attempts to match agents of type $x$ to agents of type $y$ with probability $\epsilon$ will actually generate such a match with probability no greater than $\epsilon^2$. In order to actually achieve the $\epsilon$ fraction that we desire, we must scale $\alpha_{xy}$ by $1/\epsilon$, or $\frac{\mu_{x}}{\lambda_{x}}$. Intuitively, we have scaled up the match probability according to the probability that $x$ is present, in order to achieve the rate recommended by LP-UB.  This motivates our choice of scaling factor on line~\ref{alg.line5}.

The algorithm actually scales the probability by an additional factor of $\gamma$, which is a tunable parameter of the algorithm.  This is to ensure that each agent has a constant probability of being available in the system unmatched when its ideal match arrives.  We will optimize $\gamma$ as part of our analysis of the algorithm.

\subsection{Analysis}

In this section we bound the competitive ratio of Algorithm \textsc{OnlineMatch}.

\begin{theorem}
\label{thm:main}
Algorithm~\ref{alg.avail} is a 8-approximation to the value of LP-UB.
\end{theorem}


\paragraph{Proof Strategy.}
Our strategy for proving Theorem~\ref{thm:main} is to show that, whenever a node of type $y \in X$ arrives to the market, it will match with a node of type $x \in X$ with probability at least $\alpha_{xy}/8$.  If this is true for each $x$ and $y$, then taking a sum over all types (multiplied by the rates at which they arrive) yields the desired 8-approximation. Note that since our objective is the long-run average match value, we will focus on event probabilities in the long-run steady-state of Algorithm~\ref{alg.avail}.

To show the bound of $\alpha_{xy}/8$ we consider a sequence of events that lead to a newly-arrived agent of type $y$ matching to an agent of type $x$.  First, \textsc{OnlineMatch} must reach the iteration of its main loop corresponding to type $x$, meaning that the agent of type $y$ does not match to a previously-considered type.  Next, there must be an agent of type $x$ that has arrived previously but (a) has not yet departed, and (b) has not yet been previously matched.  Finally, the node of type $y$ must choose to actually match to the node of type $x$, passing the probabilistic check on Line~\ref{alg.line5} of \textsc{OnlineMatch}.

Our main challenge in the analysis of \textsc{OnlineMatch} is that the events described above are correlated with each other and with the state of the market. Whether a certain type of agent is available in the market to be matched at time $t$ depends on the types of other agents present in the market, as this influences the probability that they themselves have previously matched.  Thus, the availability of different types of agents are correlated through the pool of agents waiting to be matched at any given time.  

We address this difficulty by coupling the events described above with Poisson point processes that are designed to be independent or (in one case) positively correlated with each other.
%
That is, while agents in the market are matched at rates that vary over time with the composition of available agents, these rates are subject to uniform upper and lower bounds that reflect maximum and minimum possible matching rates.  By relating to these extreme matching scenarios, we can derive uniform bounds on the success rate of matching attempts under arbitrary market conditions.

\paragraph{Defining the Poisson Point Processes.}
We begin by introducing the notion of an agent being \emph{present} in the market, and bounding the probability that a node of a given type is present at any given time.  We will say an agent $i$ is present at time $t$ if it has arrived but not yet departed; that is, if $a_i \leq t < d_i$.  We'll say the node is \emph{available} at time $t$ if it is present and has not yet been matched to another node.

Importantly, an agent can be present but not available: even after an agent has been matched, one could simulate the departure process for that agent as though they had not matched, and we view the agent as being present until they leave under that simulated process.  The advantage of considering presence, rather than availability, is that whether an agent is present at a given time depends only on their arrival and departure times, and is independent of all other agents in the market.


\begin{lemma}
\label{lem:present}
Choose a type $x \in X$.  Then over all randomness in arrivals and departures, the steady-state probability that at least one agent of type $x$ is present is $1 - e^{-\lambda_x / \mu_x}$, which is at most $\min\{\lambda_x / \mu_x, 1\}$.
\end{lemma}
\begin{proof}
Consider the Markov chain with states $\{S_0, S_1, S_2, \dotsc\}$ where $S_i$ corresponds to exactly $i$ agents of type $x$ being present.  Then the market transitions from state $S_i$ to state $S_{i+1}$ at rate $\lambda_x$, and (for each $i \geq 1$) from state $S_i$ to state $S_{i-1}$ at rate $i \cdot \mu_x$.  Write $\pi_i$ for the steady-state probability of being in state $S_i$.  Then the balance equations give that, for all $i \geq 1$, 
\[ (\lambda_x + i \cdot \mu_x)\pi_i = \lambda_x \cdot \pi_{i-1} + (i+1)\mu_x \cdot \pi_{i+1}. \]
Solving in terms of $\pi_0$ yields $\pi_i = \left(\frac{\lambda_x}{\mu_x}\right)^i\cdot\frac{1}{i!}\cdot \pi_0$ for each $i \geq 1$.  Since $\sum_{i \geq 0} \pi_i = 1$, we conclude that $1 = \pi_0 \left( \sum_{i \geq 0} \frac{(\lambda_x / \mu_x)^i}{i!}\right) = \pi_0 \cdot e^{\lambda_x / \mu_x}$.  Thus $\pi_0 = e^{-\lambda_x / \mu_x}$ and hence $1 - \pi_0 = 1 - e^{-\lambda_x / \mu_x}$ as required.
\end{proof}

We will write $P_x(t)$ for the event that at least one agent of type $x$ is present at time $t$.  Then $\neg P_x(t)$ is the event that no agents of type $x$ are present at time $t$.  And to disambiguate what happens at the moment an agent arrives: when an agent $i$ of type $x$ arrives to the market at time $t$ and we run \textsc{OnlineMatch} for that agent, we take $P_x(t)$ to mean the presence of agents of type $x$ other than $i$.  In particular, it's possible for $P_x(t)$ to be false during that execution of \textsc{OnlineMatch}.

We make the following observations about presence events:
\begin{enumerate}
    \item $P_x(t_1)$ and $P_y(t_2)$ are independent events for each pair of types $x \neq y$ and any pair of times $t_1 \leq t_2$.  
    \item For each $t_1 < t_2$, $P_x(t_1)$ and $P_x(t_2)$ are positively correlated.
    \item For each $t_1 < t_2 < t_3$, $P_x(t_1)$ and $P_x(t_3)$ are conditionally independent given $P_x(t_2)$.
\end{enumerate}

%
We next define some terminology about agents matching with each other.
We'll say that agent of type $y$ \emph{considers} matching to an agent of type $x$ if, during the execution of \textsc{OnlineMatch}, we enter the iteration of the loop on line~\ref{alg.line3} corresponding to type $x$.  We'll say that the agent \emph{attempts} to match to an agent of type $x$ if, in addition to considering the match, the probabilistic match on line~\ref{alg.line5} \emph{would} occur (regardless of whether or not the condition on line~\ref{alg.line4} evaluates to true).  In other word, we can imagine pre-evaluating the probabilistic check on line~\ref{alg.line5} before checking the condition on line~\ref{alg.line4}, and an attempted match corresponds to iterations in which the probabilistic check passes.  The attempted match is \emph{successful} if, in addition, at least one node of type $x$ is present and available.

We are now ready to analyze the presence and availability of agents in the market. Consider the following events, each of which follows a point process.  

\begin{itemize}
    \item Event $Z^1_x$: An agent of type $x$ arrives and, for every type $z$ that is currently present in the market, the agent does not attempt to match to type $z$.
    \item Event $Z^2_x$: When there is at least one agent of type $x$ present in the market, $Z^2_x$ is the event that an agent arrives and (by pre-evaluating its probabilistic choice on Line~\ref{alg.line5}) we see that this agent would attempt to match to an agent of type $x$ if it reaches that iteration of the main loop.  Otherwise, when there is no agent of type $x$ present in the market, $Z^2_x$ follows an independent Poisson clock of rate $\gamma \sum_{y \in X} \lambda_y \alpha_{xy} \max(1, \mu_x / \lambda_x)$.
    \item Event $Z^3_x$: When there is exactly one agent of type $x$ present in the market, $Z^3_x$ is the event that that agent departs.  Otherwise, when there is not exactly one agent of type $x$, $Z^3_x$ follows an independent Poisson clock of rate $\mu_x$.
    \item Event $Z^4_{x,y}$: An agent of type $y$ arrives and it attempts to match to an agent of type $x$, and moreover it did not attempt to match to any other type that is present before attempting to match to type $x$.
\end{itemize}

To motivate our interest in these events, let $A_{x,y}$ be the following aggregate event.  We say event $A_{x,y}$ occurs at time $t$ if $Z^4_{x,y}$ occurs at time $t$, and moreover there is some $t' < t$ such that $Z^1_x$ occurs at time $t'$ and neither $Z^2_x$ nor $Z^3_x$ occurs at times between $t'$ and $t$.  We claim that whenever event $A_{x,y}$ occurs, an agent of type $y$ matches to an agent of type $x$.  

\begin{lemma}
\label{lem:aggregate.event}
For each occurrence of event $A_{x,y}$ as defined above, an agent of type $y$ matches to an agent of type $x$ who arrived previously.
\end{lemma}
\begin{proof}
Say $A_{x,y}$ occurs at time $t$.  This implies event $Z^4_{x,y}$ occurs at time $t$, so an agent of type $y$ arrives at time $t$ and attempts to match to an agent of type $x$.  At time $t'$ event $Z^1_x$ occurred, so an agent of type $x$ arrived and did not immediately match.  Since event $Z^3_x$ did not occur in $[t',t]$, we have that $P_x(t'')$ is true for all $t'' \in [t',t]$.  Therefore, since $Z^2_x$ also did not occur, no agent of type $x$ could have matched to an arriving agent between times $t'$ and $t$.  We conclude that at least one agent of type $x$ is available at time $t$, so the agent of type $y$ will successfully match to an agent of type $x$.
\end{proof}

Given Lemma~\ref{lem:aggregate.event}, our goal is to derive a lower bound on the frequency of the aggregate event $A_{x,y}$.  For this it will be helpful to analyze the rates of $Z^1_x$, $Z^2_x$, $Z^3_x$, and $Z^4_{x,y}$, as well as how they relate to each other.  The following observations follow directly from the definitions of these processes.

\begin{enumerate}
    \item The arrival or departure of an agent can trigger at most one of $Z^1_x$, $Z^2_x$, or $Z^3_x$.
    \item Given the set of types present in the market, $Z^1_x$, $Z^2_x$, $Z^3_x$, and $Z^4_{x,y}$ are all independent of the availability of agent types.
    \item The rate of $Z^1_x$ at time $t$ depends on the set of types present in the market, and is positively correlated with $\neg P_z(t)$ for each type $z$.  
    \item The rate of $Z^4_{x,y}$ at time $t$ depends on the set of types present in the market, and is positively correlated with $\neg P_z(t)$ for each $z \neq x$, but is independent of $P_x(t)$.
\end{enumerate}


As we will show below, events $Z^2_x$ and $Z^3_x$ have constant rate (i.e., they are homogeneous).  The rate of $Z^1_x$ and $Z^4_{x,y}$ both depend on which agent types are present in the market, but how are they correlated with each other?  As it turns out, $Z^1_x$ and $Z^4_{x,y}$ are positively correlated, meaning that occurrences of $Z^4_{x,y}$ make it more likely to have recently seen an occurrence of $Z^1_x$.  

\begin{claim}
\label{claim:positive.correlation}
For any times $t_1 \leq t_2$, $Z^1_x(t_1)$ is weakly positively correlated with $Z^4_{x,y}(t_2)$.
\end{claim}

We prove Claim~\ref{claim:positive.correlation} in the Appendix.  The intuition is that 
given a pair of times $t_1 \leq t_2$ and the values of $P_z(t_1)$ for each $z$, the rates of $Z^1_x$ at time $t_1$ and the rate of $Z^4_{x,y}$ at time $t_2$ are conditionally independent.  This uses the fact that $Z^4_{x,y}$ does not depend on the presence of nodes of type $x$, so in particular $Z^4_{x,y}$ is not impacted by the implied presence of a node of type $x$ immediately following an event $Z^1_x$.  Thus, since $Z^1_x$ and $Z^4_{x,y}$ are both positively correlated with each $P_z(t_1)$, which are themselves independent across $z$, we conclude that the rate of $Z^1_x$ at time $t_1$ is positively correlated with the rate of $Z^4_{x,y}$ at time $t_2$.

We now bound the rates of events $Z^1_x$, $Z^2_x$, $Z^3_x$, and  $Z^4_{x,y}$.  We begin with $Z^1_x$, which captures agents of type $x$ arriving to the market and not being immediately matched. Recall that the rate of $Z^1_x$ depends on which agents are present in the market.  Lemma~\ref{lem:avail.immediate} bounds the expected rate of $Z^1_x$ at an arbitrary fixed time $t$, 
as a function of which agent types are present.

\begin{lemma}
\label{lem:avail.immediate}
Write $b_{y}(t)$ for the indicator variable for event $P_y(t)$, that an agent of type $y$ is present.
Then at any fixed time $t$, the rate of $Z^1_x$ is at least 
\[\lambda_x\left(1 - \gamma \sum_{y \in X}b_y(t) \alpha_{yx}\max\{1, \mu_y / \lambda_y\}\right).\]
\end{lemma}
\begin{proof}
Agents of type $x$ arrive at rate $\lambda_x$.  Suppose agent $i$ of type $x$ arrives at time $t$.  
For each $y \in X$, if $y$ is present then Algorithm \textsc{OnlineMatch} considers a match with type $y$ (and hence successfully matches) with probability at most $\gamma \alpha_{yx} \max\{1, \mu_y / \lambda_y \}$. 
The total probability that agent $i$ attempts to match to any other agent that is present at time $t$ is therefore at most 
\[ \gamma \sum_{y \in X}b_y(t) \alpha_{yx}\max\{1, \mu_y / \lambda_y\} \]
and hence the probability that agent $i$ does not attempt such a match is at least
\[ 1 - \gamma \sum_{y \in X}b_y(t) \alpha_{yx}\max\{1, \mu_y / \lambda_y\}. \]
%
\end{proof}

We next bound the rate of process $Z^2_x$, which recall is a superset of all events where an agent matches to a previously-arriving agent of type $x$.  

\begin{lemma}
\label{lem:avail.delayed}
The rate of $Z^2_x$ is precisely $\gamma \sum_{y \in X} \lambda_y \alpha_{xy} \max(1, \mu_x / \lambda_x)$ at all times $t$.
\end{lemma}
\begin{proof}
When no agent of type $x$ is present in the market, $z_2$ is defined to follow an independent Poisson clock with the specified rate.  So suppose an agent of type $x$ is present.
Agents of type $y$ arrive at rate $\lambda_y$. Consider an agent $i$ of type $y$ that arrives at time $t$.  
Event $Z^2_x$ is then equivalent to the event that agent $i$ \emph{would} attempt to match to type $x$ given that type $x$ is the first type considered.  This event is independent of the set of types present in the market.  Moreover, 
in this case the match occurs with probability $\gamma \alpha_{xy} \max(1, \mu_x / \lambda_x)$.
%
%
Summing over all types $y \in X$ completes the proof.
\end{proof}

It is immediate from the definition that the rate of $Z^3_x$ is exactly $\mu_x$ at all times.  So it remains to bound the rate of event $Z^4_{x,y}$, which is a subset of scenarios where an agent of type $y$ arrives and attempts to match to an agent of type $x$.  As with $Z^1_x$, we bound the rate of $Z^4_{x,y}$ at an arbitrary fixed time $t$ in expectation over randomness in agent arrivals and departures.

\begin{lemma}
\label{lem:attempt.match}
The rate of $Z^4_{x,y}$ is at least $\lambda_y (1 - \tfrac{\gamma}{2}) \cdot \gamma \cdot \alpha_{xy} \cdot \max(1, \mu_x/\lambda_x)$ at any fixed time $t$, taking expectations over the events $P_z(t)$ for all $z$.
\end{lemma}
\begin{proof}
Agents of type $y$ arrive at rate $\lambda_y$.  Suppose an agent of type $y$ arrives at time $t$, and consider the evaluation of \textsc{OnlineMatch} on this agent.  
%

We will first bound the probability that agent $i$ considers matching to an agent of type $x$.  Suppose that, when iterating over types to consider, agent $i$ will terminate the search after its first attempted match with a type that is present, whether or not the match is successful.  This is equivalent to the process described as $Z^4_{x,y}$, since an attempted match with a type that is present and available will be successful.  By Lemma~\ref{lem:present}, for each $z \in X$ a node of type $z$ is present at time $t$ with probability at most $\min\{\lambda_{z} / \mu_{z}, 1\}$. Thus, given that our algorithm considers a match with type $z$, this match will terminate the search with probability at most
\[ \gamma \alpha_{zy} \max\{1, \mu_{z} / \lambda_{z} \} \cdot \min\{\lambda_{z} / \mu_{z}, 1\} = \gamma \alpha_{zy}. \]
We note that this bound depends only on whether an agent of type $z$ is present in the market, not whether such an agent is available.  The total probability that agent $i$ terminates its search prematurely (before attempting all types) at time $t$ is therefore at most 
\[ \sum_{z \in X} \gamma \alpha_{zy}. \]
If we consider only half of the types $z\in X$ uniformly at random, and note that LP-UB guarantees $\sum_z\alpha_{zy}\leq1$, the probability of a match is then at most $\gamma / 2$ (where the expectation is over randomness in algorithm and over which types are chosen).  This is a bound on the probability that the algorithm matches to some other type before type $x$ is considered. 

Assuming it is considered, the match will be attempted with probability $\gamma \cdot \alpha_{xy} \cdot \max(1, \mu_x/\lambda_x)$.  Note that the conditional attempt probability is independent of whether the match is considered.  The unconditional probability that the match is attempted is therefore at least $(1 - \tfrac{\gamma}{2}) \cdot \gamma \cdot \alpha_{xy} \cdot \max(1, \mu_x/\lambda_x)$.
\end{proof}

Having now established bounds on the rates of each of these Poisson processes, we are now ready to bound the match probabilities of \textsc{OnlineMatch}.

\begin{lemma}
\label{lem:avail}
Suppose we set $\gamma = 1/2$ in Algorithm \textsc{OnlineMatch}.  Choose any $x \in X$, and suppose that an agent of type $y \in X$ arrives at time $t$.  Then \textsc{OnlineMatch} will match this agent to a node of type $x$ at time $t$ with probability at least $\alpha_{xy}/8$, where the probability is over any randomness in the algorithm and in the arrivals and departures of all other agents.
\end{lemma}
\begin{proof}
Recall that agents of type $y$ arrive at uniform rate $\lambda_y$.  By Lemma~\ref{lem:aggregate.event} it therefore suffices to show that event $A_{x,y}$ occurs at a rate of at least $\lambda_y \alpha_{xy}/8$ at each fixed time $t$, in expectation over the presence of agent types in the market.
Let $B(t)$ be the event that the most recent event before time $t$, from among $Z^1_x$, $Z^2_x$, and $Z^3_x$, is $Z^1_x$.  Then note that $A_{x,y}(t) = Z^4_{x,y}(t) \wedge B(t)$.

We claim that, at the long-run steady state of \textsc{OnlineMatch}, for any $\gamma \in [1/2, 1]$ the unconditional probability of event $B$ is at least
\begin{equation}
\label{eq.prob.case1}
\min\left\{1, \frac{\lambda_x}{\mu_x}\right\} \cdot \frac{1-\gamma}{2-\gamma}.
\end{equation}

We derive expression \eqref{eq.prob.case1} in Appendix~\ref{appendix:rate}. By Claim~\ref{claim:positive.correlation}, we have that the conditional rate of event $Z^1_x$ at each $t' < t$ is only higher (given $Z^4_{x,y}(t)$) than the unconditional rate.  Occurrences of event $B$ are therefore (weakly) positively correlated with occurrences of $Z^4_{x,y}$.  Thus, given the bound \eqref{eq.prob.case1}, the total unconditional rate of $A_{x,y}(t) = Z^4_{x,y}(t) \wedge B(t)$ is at least 
\[ \gamma(1-\gamma/2) \alpha_{xy} \max(1, \mu_x/\lambda_x) \frac{\lambda_x}{\mu_x} \frac{1-\gamma}{2-\gamma} = \gamma(1-\gamma/2)\frac{1-\gamma}{2-\gamma} \alpha_{xy}.\]

Optimizing over the choice of $\gamma$, we have that $\gamma(1-\gamma/2)\frac{1-\gamma}{2-\gamma}$ takes on its maximum value at $\gamma = 1/2$, in which case $\gamma(1-\gamma/2)\frac{1-\gamma}{2-\gamma} = 1/8$.  Thus, by setting $\gamma = 1/2$ in \textsc{OnlineMatch}, we conclude that agents of type $y$ arrive and match to agents of type $x$ at rate at least $\alpha_{xy} \lambda_y/8$ at each time $t$.  By linearity of expectation, the total value obtained by \textsc{OnlineMatch} is therefore at least $\frac{1}{8}\sum_{x,y \in X}v_{xy}\alpha_{xy}\lambda_y$, which is $1/8$ of the value of LP-UB.  We conclude that \textsc{OnlineMatch} is an 8-approximation, as required.
\end{proof}


\bibliographystyle{plain}
\bibliography{abstractBib}

\appendix

\section{Omitted Proofs}

\subsection{Positive Correlation of $Z^1_x$ and $Z^4_x$}

\newcommand{\E}[0]{E}

We prove Claim~\ref{claim:positive.correlation}, which is that $Z^1_x(t_1)$ and $Z^4_{x,y}(t_2)$ are positively correlated for each $t_1 \leq t_2$.  We first show a slightly more general result about positive correlation.

\begin{lemma}
\label{lem:correlated}
Choose $n \geq 0$ and suppose $\vec{B} = (B_1, \dotsc, B_n)$ is a sequence of $n$ independent binary events.  Suppose also that $A_1$ and $A_2$ are binary events that are conditionally independent given $\vec{B}$, and such that $A_1$ and $A_2$ are both weakly positively correlated with $B_i$ for each $i$.  Then $A_1$ and $A_2$ are weakly positively correlated.
\end{lemma}
\begin{proof}
Write $\overline{B_i} = 1 - B_i$ for convenience.
We proceed by induction on $n$.  If $n = 0$ then $A_1$ and $A_2$ are independent by assumption.  Choose $n \geq 1$.  We then have
\begin{align*}
    \E[A_1]\E[A_2]
    & = (\E[A_1B_n] + \E[A_1\overline{B_n}])(\E[A_2B_n] + \E[A_2\overline{B_n}]) \\
    & = \E[A_1B_n]\E[A_2B_n] + \E[A_1B_n]\E[A_2\overline{B_n}] \\
    & \quad + \E[A_1\overline{B_n}]\E[A_2B_n] + \E[A_1\overline{B_n}]\E[A_2\overline{B_n}] \\
    & = \left(\frac{\E[A_1B_n]\E[A_2B_n]}{\E[B_n]} - \frac{\E[\overline{B_n}]}{\E[B_n]}(\E[A_1B_n]\E[A_2B_n])\right)\\ 
    & \quad + \left(\frac{\E[A_1\overline{B_n}]\E[A_2\overline{B_n}]}{\E[\overline{B_n}]} - \frac{\E[B_n]}{\E[\overline{B_n}]}(\E[A_1\overline{B_n}]\E[A_2\overline{B_n}])\right)\\
    & \quad + \E[A_1B_n]\E[A_2\overline{B_n}] + \E[A_2B_n]\E[A_1\overline{B_n}]\\ 
    & = \frac{\E[A_1B_n]\E[A_2B_n]}{\E[B_n]} + \frac{\E[A_1\overline{B_n}]\E[A_2\overline{B_n}]}{\E[\overline{B_n}]} \\
    & \quad - \E[\overline{B_n}]\left( \frac{\E[A_1B_n]\E[A_2B_n]}{\E[B_n]} - \frac{\E[A_1\overline{B_n}]\E[A_2B_n]}{\E[\overline{B_n}]} \right)\\
    & \quad + \E[B_n]\left( \frac{\E[A_1B_n]\E[A_2\overline{B_n}]}{\E[B_n]} - \frac{\E[A_1\overline{B_n}]\E[A_2\overline{B_n}]}{\E[\overline{B_n}]} \right)\\
    & = \frac{\E[A_1B_n]\E[A_2B_n]}{\E[B_n]} + \frac{\E[A_1\overline{B_n}]\E[A_2\overline{B_n}]}{\E[\overline{B_n}]} \\
    & \quad - \E[B_n]\E[\overline{B_n}]\left( \E[A_1|B_n]-\E[A_1|\overline{B_n}]\right)\left( \E[A_2|B_n]-\E[A_2|\overline{B_n}]\right)\\
    & \leq \frac{\E[A_1B_n]\E[A_2B_n]}{\E[B_n]} + \frac{\E[A_1\overline{B_n}]\E[A_2\overline{B_n}]}{\E[\overline{B_n}]}\\
    & = \E[B_n]\E[A_1|B_n]\E[A_2|B_n] + \E[\overline{B_n}][A_1|\overline{B_n}]\E[A_2|\overline{B_n}]
\end{align*}
where the inequality follows because $A_1$ and $A_2$ are both positively correlated with $B_n$.  But by induction, $A_1$ and $A_2$ are weakly positively correlated with each other given $B_n$ or given $\overline{B_n}$.  Therefore
\begin{align*}
    \E[A_1]\E[A_2] & \leq \E[B_n]\E[A_1|B_n]\E[A_2|B_n] + \E[\overline{B_n}][A_1|\overline{B_n}]\E[A_2|\overline{B_n}]\\
    & \leq \E[A_1A_2B_n] + \E[A_1A_2\overline{B_n}]\\
    & = \E[A_1A_2]
\end{align*}
as required.
\end{proof}

We're now ready to prove Claim~\ref{claim:positive.correlation}, that the event $Z^1_x(t_1)$ is positively correlated with the event $Z^4_{x,y}(t_2)$ for any times $t_1 \leq t_2$.  We'll use Lemma~\ref{lem:correlated} above with $A_1 = Z^1_x(t_1)$, $A_2 = Z^4_{x,y}(t_2)$, and $B_x = P_x(t_1)$ for all $x \in X$.  Then indeed $Z^1_x(t_1)$ is positively correlated with each $B_x$.  This is true for $Z^4_{x,y}(t_2)$ as well: a type being present at time $t_1$ can only increase the probability that the type is present at time $t_2$, which is positively correlated with $Z^4_{x,y}(t_2)$.  Also, $Z^1_x(t_1)$ and $Z^4_{x,y}(t_2)$ are conditionally independent given $\vec{B}$.  This is because $\vec{B}$ contains all information about the state of which types are present at the point where $Z^1_x(t_1)$ is determined, and the only subsequent impact of the occurrence of $Z^1_x(t_1)$ is on the presence of nodes of type $x$, but $Z^4_{x,y}(t_2)$ is independent of the presence of agents of type $x$.  We conclude from Lemma~\ref{lem:correlated} above that $Z^1_x(t_1)$ and $Z^4_{x,y}(t_2)$ are (weakly) positively correlated.

\subsection{Relative rate of event $Z^1_x$ versus $Z^2_x$ and $Z^3_x$}
\label{appendix:rate}

We now derive \eqref{eq.prob.case1}, which is a bound on the unconditional probability of event $B$ from the proof of Lemma~\ref{lem:avail}.  Recall that $B(t)$ is the event that the most recent event before time $t$, from among $Z^1_x$, $Z^2_x$, and $Z^3_x$, is $Z^1_x$.

As in Lemma~\ref{lem:avail.immediate}, we'll write $b_{y}(t)$ for the indicator variable for event $P_y(t)$, that an agent of type $y$ is present.  Write $\vec{b}(t) = (b_y(t))_{y \in X}$ for the profile of such indicator variables.  We'll write
\[ F(\vec{b}) = \lambda_x \prod_y (1 - \gamma \alpha_{yx} \max\{1,\mu_y/\lambda_y\})^{b_y}\]
for the rate of event $Z^1_x$ when $\vec{b}(t) = \vec{b}$.  We'll also write
\[ R = \gamma \sum_{y \in X}\lambda_y \alpha_{xy} \max(1, \mu_x / \lambda_x) + \mu_x \]
for the sum of the rates of $Z^2_x$ and $Z^3_x$ (using the bound from Lemma~\ref{lem:avail.delayed}).  By Fact~\ref{fact:recent}, plus the fact that $Z^1_x$, $Z^2_x$, and $Z^3_x$ are disjoint and independent processes (since $Z^2_x$ and $Z^3_x$ have constant rate), we have that if the indicators $\vec{b}$ are constant, then the unconditional probability of event $B$ is at least
\begin{equation*}
\frac{F(\vec{b})}{R + F(\vec{b})}.
\end{equation*}

Write $\beta_y = \lambda_y / \mu_y$ for convenience.
By Lemma~\ref{lem:present}, the steady-state probability that $\vec{b}(t) = \vec{b}$ is 
\[ \Pr[\vec{b}] = \prod_{y: b_y=1} (1-e^{-\beta_y})\prod_{y: b_y=0}e^{-\beta_y}.\]
We therefore have that the unconditional probability of event $B$, over randomness in $\vec{b}(t)$, is
\begin{equation}
\label{eq:sum_of_rates}
\sum_{\vec{b}} \Pr[\vec{b}] \cdot \frac{F(\vec{b})}{R + F(\vec{b})}.
\end{equation}
To bound this rate, we first derive a lower bound on $F(\vec{b})$.  For each $\vec{b}$, $F(\vec{b})$ is weakly decreasing in $\alpha_{yx}$ for each $y$.  Since we know $\alpha_{yx} \leq \lambda_y / \mu_y = \beta_y$ for all $x$ and $y$ (from Condition~\eqref{eq.LP.cond1} of LP-UB), the worst case (lower bound) for $F(\vec{b})$ occurs when $\alpha_{yx} = \beta_y$ and hence $F(\vec{b}) \geq \lambda_x \prod_y (1 - \gamma)^{b_y}$.

Next we derive an upper bound on $R$.  Recall that from Condition~\eqref{eq.LP.cond2} of LP-UB, $\sum_{y \in X}\lambda_y \alpha_{xy}
+ \sum_{y \in X}\lambda_x \alpha_{yx} \leq \lambda_x$.  Thus $\sum_{y \in X}\lambda_y \alpha_{xy} \leq \lambda_x(1 - \sum_{y \in X}\alpha_{yx}) = \lambda_x(1 - \sum_{y \in X}\beta_y)$.  We conclude that $R \leq \gamma \lambda_x \max(1, \mu_x / \lambda_x) (1 - \sum_y \beta_y) + \mu_x$.

We now consider two cases, based on which of $\mu_x$ or $\lambda_x$ is larger.

\noindent
\textbf{Case 1: $\mu_x \leq \lambda_x$.}  Then $\max\{1, \mu_x / \lambda_x\} = 1$.  Substituting our bounds on $F(\vec{b})$ and $R$ into \eqref{eq:sum_of_rates} above, we have that the unconditional probability of event $B$ is at least
\begin{equation*}
\sum_{\vec{b}} \Pr[\vec{b}] \cdot \frac{\lambda_x \prod_{y \in X}(1-\gamma)^{b_y}}{\gamma \lambda_x (1 - \sum_y \beta_y) + \mu_x + \lambda_x \prod_{y \in X}(1-\gamma)^{b_y}}.
\end{equation*}
Since $\mu_x \leq \lambda_x$ by assumption, we can divide top and bottom by $\lambda_x$ to see that the probability is at least
\begin{equation}
\label{eq:sum_case1}
\min\left\{1, \frac{\lambda_x}{\mu_x}\right\} \cdot \sum_{\vec{b}} \Pr[\vec{b}] \cdot \frac{\prod_{y \in X}(1-\gamma)^{b_y}}{\gamma (1 - \sum_y \beta_y) + 1 + \prod_{y \in X}(1-\gamma)^{b_y}}
\end{equation}
where the extra term $\min\left\{1, \frac{\lambda_x}{\mu_x}\right\}$ is simply equal to $1$ in this case.

\noindent
\textbf{Case 2: $\mu_x > \lambda_x$.}  Then $\max\{1, \mu_x / \lambda_x\} = \mu_x / \lambda_x$.  Substituting our bounds on $F(\vec{b})$ and $R$ into \eqref{eq:sum_of_rates} above, we have that the unconditional probability of event $B$ is at least
\begin{equation*}
\sum_{\vec{b}} \Pr[\vec{b}] \cdot \frac{\lambda_x \prod_{y \in X}(1-\gamma)^{b_y}}{\gamma \lambda_x (\mu_x / \lambda_x) (1 - \sum_y \beta_y) + \mu_x + \lambda_x \prod_{y \in X}(1-\gamma)^{b_y}}.
\end{equation*}
Since $\mu_x > \lambda_x$ by assumption, this expression is at least we can divide top and bottom by $\mu_x$ to again see that the probability is at least 
\begin{equation*}
\sum_{\vec{b}} \Pr[\vec{b}] \cdot \frac{\lambda_x \prod_{y \in X}(1-\gamma)^{b_y}}{\gamma \mu_x (1 - \sum_y \beta_y) + \mu_x + \mu_x \prod_{y \in X}(1-\gamma)^{b_y}}.
\end{equation*}
Since $\lambda_x/\mu_x < 1$, we can pull out a factor of $\min\{1, \lambda_x/\mu_x\} = \lambda_x / \mu_x$ from each term to obtain \eqref{eq:sum_case1}.

This ends the case analysis.  In each case, the probability of event $B$ is at least \eqref{eq:sum_case1}.  We will now derive a lower bound on this probability.  For convenience we'll omit the leading coefficient $\min\left\{1, \frac{\lambda_x}{\mu_x}\right\}$ in the calculations below, and focus on the summation within \eqref{eq:sum_case1}.  We will bound this sum by considering only summands in which $0$, $1$, or $2$ elements of vector $\vec{b}$ is equal to $1$.  (This restriction is related to the fact that we will eventually choose $\gamma = 1/2$).  Using our expression for $\Pr[\vec{b}]$, and defining $\beta = \sum_y \beta_y$, we can rewrite these terms of \eqref{eq:sum_case1} as
\begin{align*}
& e^{-\beta}\frac{1}{\gamma(1-\beta) + 2} + \sum_y (1-e^{-\beta_y})e^{\beta-\beta_y}\frac{(1-\gamma)}{\gamma(1-\beta) + 2 - \gamma}\\
& + \sum_{y_1, y_2}(1-e^{-\beta_{y_1}})(1-e^{-\beta_{y_2}})e^{\beta-\beta_{y_1}-\beta_{y_2}}\frac{(1-\gamma)^2}{\gamma(1-\beta) + 1 + (1 - \gamma)^2}.
\end{align*}
This expression is convex in each $\beta_y$ and weakly decreasing in $\beta$, so its minimum occurs when $\sum_y \beta_y = \sum_y \alpha_{yx} \leq 1$ and all $\beta_y$ are equal.  Substituting into the expression above and using the fact that $n(e^{1/n} - 1) \geq 1$ for all $n \geq 1$, we have that the probability of event $B$ is at least
\[ \frac{1}{e} \cdot \frac{1}{2} + \frac{1}{e} \cdot \frac{1-\gamma}{2-\gamma} + \frac{1}{2e} \cdot \frac{(1-\gamma)^2}{1+(1-\gamma)^2} \]
which is greater than $\frac{1-\gamma}{2-\gamma}$ for all $\gamma \in [1/2, 1]$.  We conclude that \eqref{eq:sum_case1} is at least
\[ \min\left\{1, \frac{\lambda_x}{\mu_x}\right\} \cdot \frac{1-\gamma}{2-\gamma}\]
for all $\gamma \in [1/2, 1]$, as required.
\end{document}